\documentclass[10pt,twocolumn,letterpaper]{article}

\usepackage[pagenumbers]{iccv} 

\definecolor{iccvblue}{rgb}{0.21,0.49,0.74}
\usepackage[pagebackref,breaklinks,colorlinks,allcolors=iccvblue]{hyperref}
\usepackage{graphicx}
\usepackage{subcaption}
\usepackage{booktabs}
\usepackage{amsmath}
\usepackage{amssymb}
\usepackage{mathtools}
\usepackage{amsthm}
\usepackage{yhmath}
\usepackage{multirow}
\usepackage[capitalize,noabbrev]{cleveref}
\usepackage{algorithm, algorithmic}
\usepackage[accsupp]{axessibility}

\theoremstyle{plain}
\newtheorem{theorem}{Theorem}[section]
\newtheorem{proposition}[theorem]{Proposition}

\theoremstyle{definition}

\theoremstyle{remark}

\def\paperID{1772} 
\def\confName{ICCV}
\def\confYear{2025}

\title{Learning Null Geodesics for Gravitational Lensing Rendering\\ in General Relativity}

\author{Mingyuan Sun\textsuperscript{\rm 1}\quad Zheng Fang\textsuperscript{\rm 1,}\footnotemark[2]\quad Jiaxu Wang\textsuperscript{\rm 2}\quad Kunyi Zhang\textsuperscript{\rm 3}\quad Qiang Zhang\textsuperscript{\rm 2,4}\quad Renjing Xu\textsuperscript{\rm 2,}\setcounter{footnote}{1}\thanks{Corresponding authors.}\\
\textsuperscript{\rm 1}Northeastern University \quad
\textsuperscript{\rm 2}The Hong Kong University of Science and Technology (Guangzhou)\\
\textsuperscript{\rm 3}Sichuan University\quad
\textsuperscript{\rm 4}Beijing Innovation Center of Humanoid Robotics Co., Ltd.\\\\
{\small\href{https://myuansun.github.io/gravlensx}{\textcolor[HTML]{ED008A}{\texttt{myuansun.github.io/gravlensx}}}}
}

\begin{document}
\maketitle
\begin{abstract}
We present \texttt{GravLensX}, an innovative method for rendering black holes with gravitational lensing effects using neural networks. The methodology involves training neural networks to fit the spacetime around black holes and then employing these trained models to generate the path of light rays affected by gravitational lensing. This enables efficient and scalable simulations of black holes with optically thin accretion disks, significantly decreasing the time required for rendering compared to traditional methods. We validate our approach through extensive rendering of multiple black hole systems with superposed Kerr metric, demonstrating its capability to produce accurate visualizations with significantly $15\times$ reduced computational time. Our findings suggest that neural networks offer a promising alternative for rendering complex astrophysical phenomena, potentially paving a new path to astronomical visualization.
\end{abstract}

\vspace{-20pt}
\section{Introduction}
\label{sec:intro}
\vspace{-3pt}
General relativity~\citep{einstein1916grundlage}, formulated by Albert Einstein in 1915, fundamentally altered our understanding of gravity. Prior to this theory, gravity was viewed through the lens of Newtonian mechanics as a force acting at a distance between two masses. However, Einstein's theory reframed gravity as the warping of space and time, collectively known as spacetime, by mass and energy. In this framework, massive objects like stars and planets cause spacetime to curve, and this curvature dictates the motion of objects, including the paths of light rays. Einstein's field equations, the mathematical backbone of general relativity, describe how matter and energy determine the curvature of spacetime, providing a profound insight into the dynamics of the universe at both cosmic and quantum scales.

\begin{figure}[!t]
    \centering
    \includegraphics[width=\linewidth]{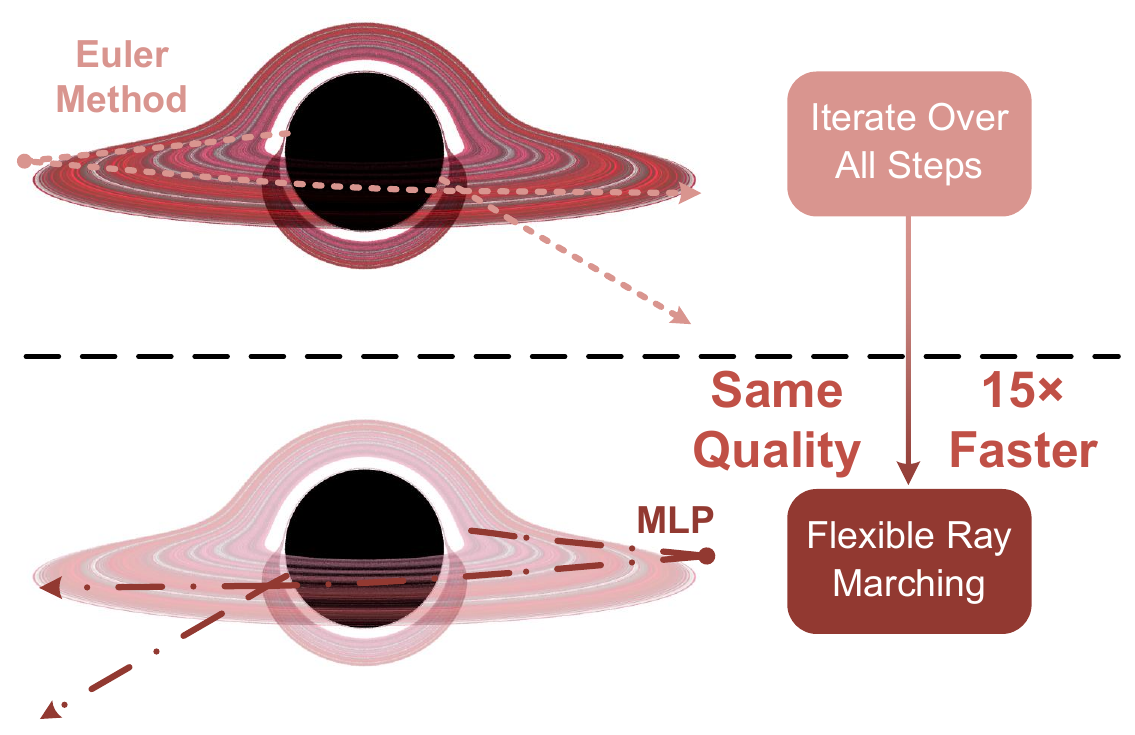}
    \caption{\textbf{(Top)} Traditional techniques for rendering black holes typically solve for geodesics through iterative Runge-Kutta method. \textbf{(Bottom)} Our approach, however, can determine any position of light rays in curved spacetime with just one forward pass, significantly lowering computational costs and enhancing adaptability to diverse point sampling methods in the rendering field.}\vspace{-12pt}
    \label{fig:intro}
\end{figure}

Gravitational lensing is one of the most compelling visual consequences of general relativity. When light from a distant object, such as a galaxy or a star, passes near a massive object, the curvature of spacetime bends the light's path. This bending can magnify, distort, or create multiple images of the background object, a phenomenon analogous to the way a glass lens bends and focuses light. Gravitational lensing has become a powerful tool in modern astronomy, allowing scientists to study objects too faint or too distant to be observed directly. Moreover, it provides a way to detect dark matter, measure the distribution of mass in the universe, and even observe exoplanets in distant star systems.

Gravitational lensing is especially dramatic around black holes due to their immense gravitational pull. When light passes close to a black hole, it can be bent to such an extent that it creates some of the most spectacular lensing effects observed in the universe. The movie Interstellar (2014), directed by Christopher Nolan, brought the depiction of black holes into popular culture in a visually stunning and scientifically informed manner, which is analyzed and reproduced by \citet{james2015gravitational}. To visualize a black hole with gravitational lensing effect, we need to solve the field equations of general relativity to determine the paths of light rays as they travel through the warped spacetime near the black hole. Despite the complexity of the Einstein field equations, there are several metrics that introduce additional constraints to simplify the spacetime around a single black hole, such as Schwarzschild metric, Kerr metric, Reissner-Nordstr\"om metric, etc. 

The rendering process of a black hole typically begins by defining a viewpoint in space, known as the camera's position. Subsequently, a virtual image plane is established in front of this viewpoint, and the initial directions of light rays are generated by casting rays from the camera through the image plane. These rays are iteratively traced through the spacetime over a series of discrete steps according to the specific metric, making the whole process extremely computationally expensive. \citet{james2015gravitational} leverages Schwarzschild metric in spherical coordinates to render a single black hole and simplifies the geodesic equation by an ingenious affine transformation. This approach significantly reduces the computational cost of ray tracing, making it feasible to render single black hole in real-time. \citet{gralla2020null} derive a semi-analytic solutions to the null geodesic equations in Kerr spacetime via elliptic integrals, enabling dramatically faster ray tracing for rotating black holes. However, when it comes to rotating black hole or multiple black holes, this simplification is no longer applicable, and the rendering process is still computationally expensive.

\vspace{-11pt}
\paragraph{Contribution} In this work, we propose a novel approach to render black holes with gravitational lensing effect by incorporating neural networks into ray tracing, namely \texttt{GravLensX}. Here, Neural networks are trained to fit the spacetime around black holes and conduct ray tracing, significantly reducing the computational cost of rendering. We demonstrate the effectiveness of our approach by rendering multiple black holes with optically thin accretion disks in Kerr metric. Though our approach cannot achieve real-time rendering, it reduces the computation time of rendering by an order of magnitude compared to traditional methods. This work opens up new possibilities for rendering black holes with lightweight neural networks, enabling more efficient and scalable simulations of black holes in various scenarios.
\vspace{-6pt}
\section{Related Work}
\subsection{Black Hole Rendering}
In general relativity, the curvature of spacetime is determined by the distribution of matter and energy, causing light to bend as it passes near massive objects like black holes—a phenomenon known as gravitational lensing. Solving the light path of photons near a black hole is one the most crucial part of black hole rendering. A black hole can be uniquely characterized by three fundamental properties: its mass, electric charge, and angular momentum. The Schwarzschild metric \citep{schwarzschild1916gravitationsfeld} is the simplest solution to the Einstein field equations, describing a non-rotating black hole characterized by its mass. A generalization of this solution is the Kerr metric \citep{kerr1963gravitational}, which accounts for a rotating black hole with both mass and angular momentum. \citet{newman1965metric} gives the solution to the Einstein-Maxwell equations in general relativity, which comprehensively describe the gravitational field outside a rotating, charged black hole.

Accretion flows of black holes are often classified by optical depth into two regimes. Optically thin disks can be rendered efficiently using semi-analytic geodesics, whereas optically thick disks demand a full radiative-transfer solution—including absorption, scattering, and Faraday rotation—typically via specialized codes such as ipole \citep{moscibrodzka2018ipole}. In this work, we do not model the disk's detailed optical or geometric properties, treating it as a simple thin texture.

Recent studies have employed the Schwarzschild metric \citep{riazuelo2019seeing, meseguer2023custom} and the Kerr metric \citep{james2015gravitational, meseguer2023custom}, in combination with the Runge-Kutta method, to numerically solve the trajectory of photons in the vicinity of a single black hole for rendering. Multiple black hole rendering, as a more challenging field, was first explored by \citet{bohn2015does}, which calculates the light path with SpEC~\footnote{\href{http://black-holes.org/SpEC.html}{http://black-holes.org/SpEC.html}} using the 3+1 decomposition, a general framework for describing any spacetime geometry. Since SpEC tool does not support GPU acceleration, the rendering process is quite time-consuming. As an alternative, \citet{combi2021superposed} proposed a superposed metric to approximate the spacetime of multiple black holes, and was further applied for analyzing Electromagnetic Signatures of an Analytical Mini-Disk Model~\citep{porter2024parameter}. \citet{levis2022gravitationally} introduced a gravitational lensing-based tomography approach to reconstruct black hole emission structures from multi-angle observational data. In this paper, we mainly leverage this superposed metric to construct multi-black-hole system. Although previous works prefer Runge-Kutta method to solve the geodesic equation. In our implementation, we find that the most widely used Runge-Kutta method, RK4, does not perform better than the Euler method at the same step size, but it costs significantly more time. Therefore, we adopt the Euler method to solve the geodesic equation for each light ray.
\vspace{-6pt}
\subsection{Physics-informed Neural Network} 
Physics-Informed Neural Network \citep{raissi2019physics} is a method in scientific machine learning addressing problems related to Partial Differential Equations (PDEs). PINNs integrate domain-specific knowledge in the form of physical laws, into the architecture and training of neural networks, ensuring that the outputs of the network not only fit the training data but also comply with the underlying physical principles. PINNs have been successfully applied to a wide range of problems in a wide range of fields that contain complex PDEs, such as fluid dynamics~\citep{raissi2019physics, eivazi2022physics, arthurs2021active}, solid mechanics~\citep{arthurs2021active, haghighat2021physics},  thermodynamics~\citep{cai2021physics, costabal2024delta}, refractive field~\citep{zhao2024single}, and electromagnetics~\citep{khan2022physics}. PINNs have also been used to compute black-hole quasinormal mode spectra by directly solving the linear perturbation equations around a Schwarzschild or Kerr background, including the Regge-Wheeler equation~\citep{ovgun2021quasinormal, cornell2022using} and Teukolsky equation~\citep{luna2023solving, cornell2024solving}. Regge-Wheeler equation and Teukolsky equation are respectively derived from Schwarzschild metric and Kerr metric, describe the evolution of perturbations of massless fields (gravitational waves, electromagnetic waves, or scalar fields) in the curved spacetime of a black hole. We extend the idea of PINNs to our neural network training, resulting in improved accuracy in geodesic fitting.
\vspace{-6pt}
\setlength{\textfloatsep}{5pt}
\begin{algorithm}[!t]
    \caption{Black Hole Rendering with Euler Method}
    \label{alg:euler_render}
    \renewcommand{\algorithmicrequire}{\textbf{Input:}}
    \renewcommand{\algorithmicensure}{\textbf{Output:}}
    \begin{algorithmic}[1]
        \REQUIRE $P^{\text{bh}}$, $M^{\text{bh}}$, $A^{\text{bh}}$, $l^{\text{out}}$, $l^{\text{in}}$, $\Delta\lambda$, $p^{\text{init}}$,$v^{\text{init}}$
        \ENSURE Color $C$    
        \STATE \# Initialize color, affine parameter, transmittance, etc.
        \STATE $C=[0, 0, 0]^\mathsf{T}$, $\lambda=0$, $T=1$
        \STATE $p= p_{init}$, $v= v_{init}$, $v_t=0$
        \WHILE{CheckInRegion$(p, l^{\text{out}}, l^{\text{in}})$}
            \STATE \# Update color
            \STATE $C', \sigma' = \operatorname{GetColor}(p)$
            \STATE $C\leftarrow C+C'\sigma'T\Delta\lambda$
            \STATE $T\leftarrow T\cdot\operatorname{exp}(-\sigma'\Delta\lambda)$
            \STATE \# Calculate g
            \STATE $g= I_{4 \times 4}$
            \FOR{($p^{\text{bh}}, m, a$) in ($P^{\text{bh}}, M^{\text{bh}}, A^{\text{bh}}$)}
                \STATE $(x, y, z) \leftarrow p - p^{\text{bh}}$
                \STATE $r\leftarrow\sqrt{\frac{x^2+y^2+z^2-a^2+\sqrt{\left(x^2+y^2+z^2-a^2\right)^2+4 a^2 z^2}}{2}}$
                \STATE $\ell\leftarrow\left[1, \frac{r x+a y}{r^2+a^2}, \frac{r y-a x}{r^2+a^2}, \frac{z}{r}\right]^\mathsf T$
                \STATE $g\leftarrow g+\frac{2mr^3}{r^4+a^2z^2}\ell\ell^\mathsf{T}$
            \ENDFOR
            \STATE \# Extend space coordinate to spacetime coordinate
            \STATE $u=\operatorname{Concat}(v_t, v)$, $q=\operatorname{Concat}(0, p)$, $\eta = [0, 0, 0]$
            \STATE $u_t\leftarrow$ Solve $g_{\mu \nu} u_{\mu} u_{\nu}=0$
            \FOR{$\mu=2$ to $4$}
            \STATE $\Gamma = \mathbf{0}_{4 \times 4}$
            \FOR{$\alpha, \beta = 1$ to $4$}
            \STATE $\Gamma_{\alpha\beta} = \frac{1}{2}\sum_{\mu\nu}g^{-1}_{\mu \nu}\left(\frac{\partial g_{\nu \alpha}}{\partial q^\beta}+\frac{\partial g_{\nu \beta}}{\partial q^\alpha}-\frac{\partial g_{\alpha \beta}}{\partial q^\nu}\right)$
            \ENDFOR
            \STATE $\eta[\mu-1]=-\sum_{\alpha\beta}u_{\alpha}u_{\beta} \Gamma_{\alpha\beta}$
            \ENDFOR
            \STATE \# Update position and velocity
            \STATE $p\leftarrow p+v\Delta\lambda$
            \STATE $v\leftarrow v+\eta\Delta\lambda$
            \STATE $v \leftarrow \text{Normalize}(v)$
        \ENDWHILE
    \end{algorithmic}
\end{algorithm}

\section{Method}
\vspace{-4pt}
Light goes along the geodesic in the curved spacetime near a black hole, and we propose to represent the geodesic implicitly by neural networks. In our framework, We first use Euler method to generate a series of data points along several randomly sampled geodesics (\cref{sec:bh_metric}), and then train neural networks to approximate the geodesic (\cref{sec:learn_geodesic}). Finally we build a ray tracing pipeline upon the trained neural networks to render the black holes (\cref{sec:ray_tracing}).
\vspace{-6pt}
\subsection{Classical Black Hole Rendering}
\label{sec:bh_metric}
\vspace{-4pt}
In this subsection, we introduce how we represent the spacetime around black holes and how we render images with classical method. We implement the Euler method to numerically solve the geodesic equation for each ray from the view point, which generates a series of data points along the geodesic, available for both rendering black holes and traning neural networks.

The detailed algorithm for rendering black holes with Euler method is shown in \cref{alg:euler_render}. $P^{\text{bh}}$, $M^{\text{bh}}$, $A^{\text{bh}}$ are the positions, masses, spin parameters of the black holes. We restrict the space to a spherical region centered at the origin with a radius of $l^{\text{out}}$. Once the distance of a ray to any black hole is smaller than $l^{\text{in}}$, we consider that the light ray has fallen into the black hole. $\Delta\lambda$ is the step size of the affine parameter. $p_{init}$ and $v_{init}$ are the position of the view point and the initial direction of the ray, respectively. A detailed introduction of the black hole metric we leverage is presented in Appendix B. We employ the classic volume rendering pipeline~\citep{volume_1984} to visualize the black hole that assumes a constant-speed ray for uniform spatial sampling. However, due to the variability of light speed (w.r.t. the affine parameter $\lambda$) near a black hole, evenly sampling in the affine parameter results in non-uniform spatial distribution. To address this, we normalize the ray velocity with respect to $\lambda$ at each step as can be seen in line 26 of \cref{alg:euler_render}, ensuring a constant light speed of $1$ and maintaining a direct proportionality between spatial distance along the ray and the affine parameter difference.

The rendering results of the Euler method are highly sensitive to the choice of $\Delta\lambda$. A smaller $\Delta\lambda$ improves the accuracy of the rendering but demands significantly more computational resources. Therefore, using an efficient method to solve the geodesic equation is crucial for rendering black holes effectively.
\vspace{-6pt}
\subsection{Learning Geodesics}
\label{sec:learn_geodesic}
In volume rendering, rays are casted from the viewpoint at a constant speed, and the color of each pixel is determined by the accumulated color along the ray. The ray casting equation can be written as
\vspace{-6pt}\begin{equation}
    \label{eq:ray_casting}
    r(t)=p+t\cdot d.
    \vspace{-6pt}
\end{equation}
Here, $p$ represents the position of the viewpoint, and $d$ denotes the direction of the ray, which can also be interpreted as the initial speed of light. This equation allows us to determine the position along the ray at any time $t$, facilitating more effective point sampling. When it comes to curved spacetime, the linearity of \cref{eq:ray_casting} is no longer valid, and time $t$ is no longer proper to measure the distance along the ray. Instead, we use the affine parameter $\lambda$ to parameterize the geodesic. Let us define
\vspace{-6pt}\begin{equation}
    \label{eq:ray_casting_geodesic}
    f(p, d, \lambda)=p+\int_{\lambda_0}^{\lambda}v(\lambda')\mathrm{d}\lambda',
    \vspace{-6pt}
\end{equation}
subject to $v(\lambda_0)=d$. This equation describes the position of a light ray given affine parameter $\lambda$, position $p$, and initial direction $d$. We proceed to approximate the function $f$ with a neural network in a data-driven manner. As $\lambda$ represents the distance of the output position to the input position along the curve, when $\lambda$ is large, even small variations in $p$ and $d$ can result in large changes in the output position. On the other hand, when $\lambda$ is small, small changes in $p$ and $d$ lead to only minor adjustments in the output position. This indicates that both low-frequency and high-frequency information need to be learned. Thus, we formulate $f$ as the composition of a trainable physics informed neural network $f_\theta$ and a position encoding function $\gamma$:
\vspace{-6pt}\begin{multline}
    \gamma(x) = \left(\sin(2^0 \pi x), \cos(2^0 \pi x), \ldots,\right. \\
    \left.\sin(2^{L-1} \pi x), \cos(2^{L-1} \pi x)\right),
\end{multline}
yielding a better representation of the both high and low frequency information in the input, following the Fourier feature mapping introduced by \citet{tancik2020fourier}. Then we predict the position of the light ray as
\vspace{-4pt}\begin{equation}
    \hat{p} = f_\theta\left(\gamma(p), \gamma(d), \lambda\right).
\end{equation}
The predicted position $\hat{p}$ is expected to be close to the actual position $p$ when the neural network is well-trained. We define the first loss function as the mean squared error between the predicted position and the actual position:
\vspace{-6pt}\begin{equation}
    \mathcal{L}_p = \frac{1}{N}\sum_{i=1}^{N}\left\|f_\theta\left(\gamma(p^{\text{init}}_i), \gamma(d^{\text{init}}_i), \lambda_i\right)-p^\lambda_i\right\|^2.
    \vspace{-6pt}
\end{equation}
Since our target is to approximate $f$ in \cref{eq:ray_casting_geodesic} with $f_\theta$, we also need to ensure that the velocity of the predicted position, i.e., $\frac{\mathrm{d}\hat{p}}{\mathrm{d}\lambda} = v(\lambda),$ is equivalent to the actual velocity. We define the second loss function as the mean squared error between the predicted velocity using auto differentiation and the actual velocity:
\vspace{-8pt}\begin{equation}
    \mathcal{L}_v = \frac{1}{N}\sum_{i=1}^{N}\left\|\frac{\mathrm{d}f_\theta\left(\gamma(p^{\text{init}}_i), \gamma(d^{\text{init}}_i), \lambda_i\right)}{\mathrm{d}\lambda}-v_i\right\|^2,
\end{equation}
which reveals the physics informed property of the network. The final loss function is formulated as a weighted sum of them:
\begin{equation}
    \mathcal{L} = \mathcal{L}_p + \alpha\mathcal{L}_v,
\end{equation}
where $\alpha$ is manually set as 800 to match the magnitude of $\mathcal{L}_p$. We train the network using the data produced in \cref{sec:bh_metric}. For each light ray, we sample $N_{p}$ points along the geodesic using the method described in \cref{alg:euler_render}. Each of these points consists of the initial position $p^{\text{init}}$, initial direction $v^{\text{init}}$, affine parameter $\lambda$, current position $p$, and current velocity $v$.

\begin{figure*}[!t]
    \centering
    \includegraphics[width=\linewidth]{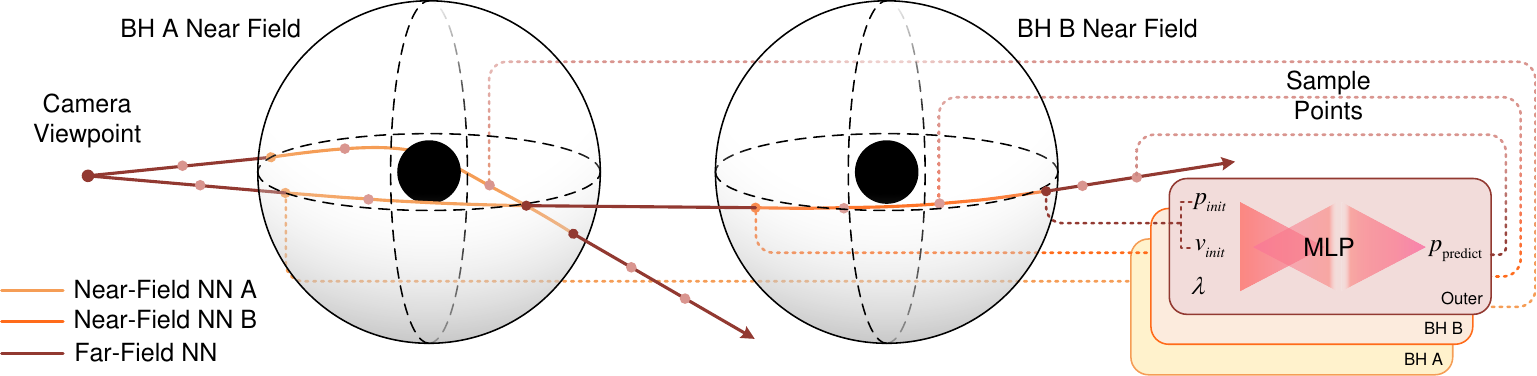}
    \caption{Our ray tracing framework with MLP. Every ray is represented by a couple of ray segments in the space, separated by the boundary of near-field regions. Each segment is represented by its starting position $p_{\text{init}}$, direction $v_{\text{init}}$ and $\lambda_{\text{end}}$.}
    \label{fig:ray_tracing}
    \vspace{-12pt}
\end{figure*}

Furthermore, we empirically find that approximating the whole spacetime with a single neural network yields poor performance. Thus, we divide the space into several regions and train a separate neural network for each region. We observed that the curvature of the spacetime near the black hole is significantly larger than that in the far field. Therefore, we divide the space into two kinds of regions: the near field and the far field. Near field refers to a spherical vicinity of the black hole, where the curvature is significant, while the far field refers to the region far away from the black hole, where the curvature is negligible. We train a neural network for each spherical vicinity of near field defined by
\vspace{-4pt}\begin{equation}
S_i = \left\{ p \mid \| p - c_i \| \leq R^{\text{bh}} + \epsilon \right\}, \quad i = 1, 2, \dots, N
\end{equation}
in which $c_i$ is the center of the $i$-th spherical vicinity, $R^{\text{bh}}$ is the spherical radius we set manually, and $N$ is the number of black holes. Also, we train the far-field region as another single neural network in a spherical region defined by 
\vspace{-6pt}\begin{multline}
    S_{\text{sky}} = \left\{ p \mid \| p \| \leq R^{\text{sky}} + \epsilon \right\}\\
    \setminus \bigcup_{i=1}^{N} \left\{ p \mid \| p - c_i \| < R^{\text{bh}} - \epsilon \right\},
\end{multline}
where $R^{\text{sky}}$ is the radius of the far field, also the radius of the sky sphere. We collect data saparately for each region and train the neural networks with the corresponding data. $\epsilon$ ensures the numerical stability near the boundaries, which is practically set to be $0.1$. Next we introduce our learning based ray tracing framework.
\vspace{-6pt}
\subsection{Ray Tracing in Curved Spacetime}
\label{sec:ray_tracing}
\vspace{-4pt}
The most outstanding advantage of representing geodesics with neural networks is that we can directly sample any point on the geodesic in one single forward propagation. This is particularly useful in volume rendering, where we can dynamically sample points along the geodesic. In our setting of rendering black holes, points with colors are concentrated on the accretion disk and the sky sphere. Thus, we propose an efficient ray tracing framework to locate points on the accretion plane and the sky sphere on all light rays casted from the viewpoint. The framework is illustrated in \cref{fig:ray_tracing}. Our rendering process can be divided into three steps: ray segments identification, color points sampling, and volume rendering.
\vspace{-12pt}
\paragraph{Ray Segments Identification} We first initialize the light rays from the viewpoint which can be placed at any position in the space. Next, we need to consider the potential events that may occur to the ray, such as intersecting the boundary of a near field, falling into a black hole~(intersecting a small-radius sphere centered at a black hole, denoted as in-black-hole boundary), or reaching the sky sphere. The goal is to quickly determine the next intersection point for the ray and update its state accordingly. For each ray, its current position and direction can be represented by its starting position $p_{\text{init}}$, $v_{\text{init}}$, an affine parameter $\lambda$, and an region index indicating the neural network belonging to which region is used. In each step, we approximate the change of the affine parameter $\Delta\lambda$ to update $\lambda$ as $\lambda \leftarrow \lambda + \Delta\lambda$, pushing the ray closer to its next intersection point while not surpassing the near-field boundary, in-black-hole boundary, or sky sphere a lot (less than $\epsilon$) to avoid inaccurate predictions of neural networks in incorrect regions. Once the ray goes into another region, we record its current state to represent a ray segment in the previous region. In \cref{sec:bh_metric} we normalize the speed of the light ray w.r.t. $\lambda$ as 1, which means that if we assume the light ray to go straight, the increment of $\lambda$ is equal to the distance between the current position and the next intersection point. Below, we provide a detailed illumination for both far-field and near-field conditions on the selection of $\Delta\lambda$. 

In far-field scenarios, Since the curvature of the spacetime is quite minor, and the path is nearly a straight line. Based on this fact, we directly cast a straight ray from the current position with its current direction, and calculate the distance $l^{\text{straight}}$ from its position to its next intersection point to a proximal boundary of the near field or sky sphere, whose radius satisfying $R^{\text{bh}}-\epsilon<R^{\text{prox}}<R^{\text{bh}}$ for near-field boundary and $R^{\text{sky}}<R^{\text{prox}}<R^{\text{bh}}+\epsilon$ for sky sphere boundary. We choose $\Delta\lambda$ as $l^{\text{straight}}$ in this case.

In near-field scenarios, the strong influence of the spacetime leads the light path to bend significantly. We further divide the near field into outside region and inside region separated by an radius of $k m$, in which $k$ is an coefficient and $m$ is the mass of the black hole. In outside region, we simplify the condition into classical Newtonian mechanics, where the geodesic represents the trajectory of a particle with initial velocity, attracted to the black hole's center. We demonstrate that with this simplification, the distance $ l^{\text{straight}} $ is always less than or equal to  $l^{\text{geodesic}}$, where $ l^{\text{straight}} $ and $ l^{\text{geodesic}} $ represent the distances from the current position to the next intersection with the proximal boundary along the straight line and the geodesic path, respectively (refer to Appendix C). This approach ensures that points outside the proximal radius are not selected, thereby maintaining neural stability. We also prove that by iterating over the geodesic in this manner, the distance between the ray point and the boundary converges to $0$. In a manner akin to far-field conditions, we choose the proximal boundary's radius such that $ R^{\text{bh}} < R^{\text{prox}} < R^{\text{bh}} + \epsilon $. We choose $\Delta\lambda$ as $l^{\text{straight}}$ as well. In inside region, the geodesic is more complex, and we simply pick $\Delta\lambda$ as the distance from the current position to the in-black-hole boundary.

A ray is terminated once it falls into a black hole or reaching the sky sphere.

\begin{figure*}[htp]  
    \centering
    \begin{subfigure}[b]{\columnwidth}  
        \centering
        \includegraphics[width=\textwidth]{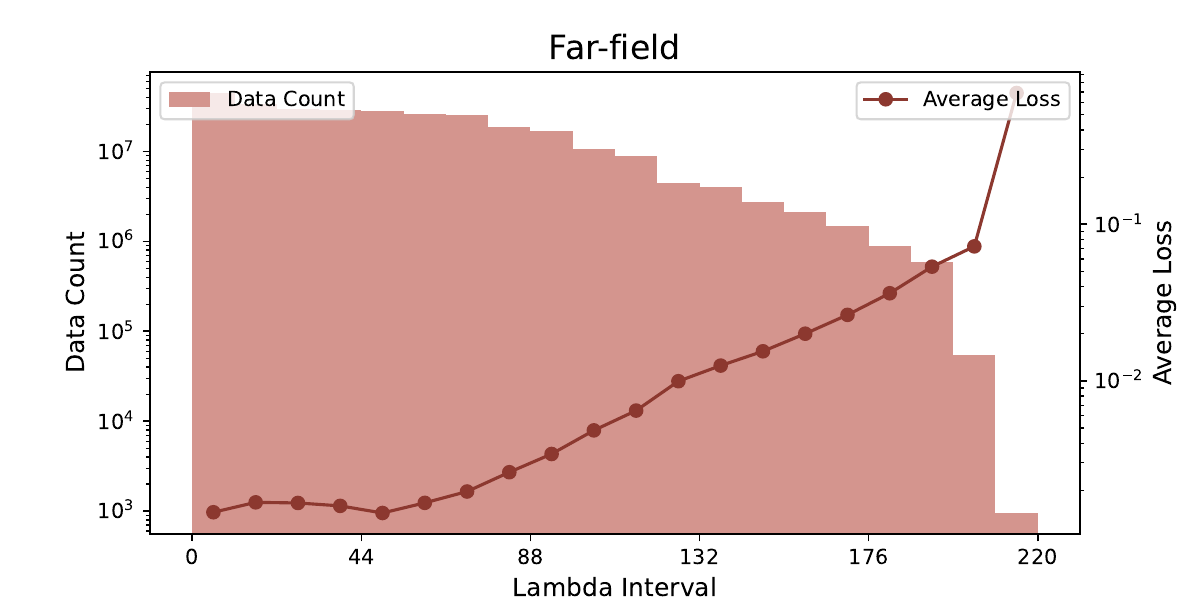}  
        \caption{}  
        \label{fig:farloss}
    \end{subfigure}
    \begin{subfigure}[b]{\columnwidth}  
        \centering
        \includegraphics[width=\textwidth]{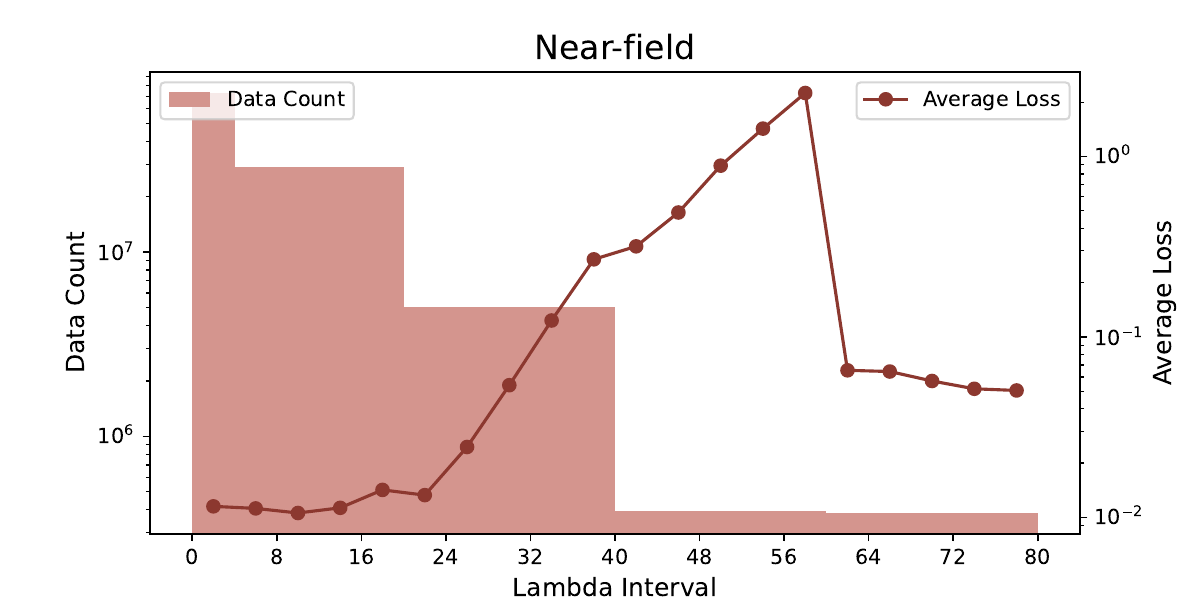}  
        \caption{}
        \label{fig:nearloss}
    \end{subfigure}
    \vspace{-10pt}
    \caption{The loss and data count distribution of near-field and far-field MLPs in 2-black-hole system.}
    \label{fig:loss_count}\vspace{-10pt}
\end{figure*}

\begin{figure*}[htp]  
    \centering
    \begin{subfigure}[b]{0.32\textwidth}  
        \centering
        \includegraphics[width=1\textwidth]{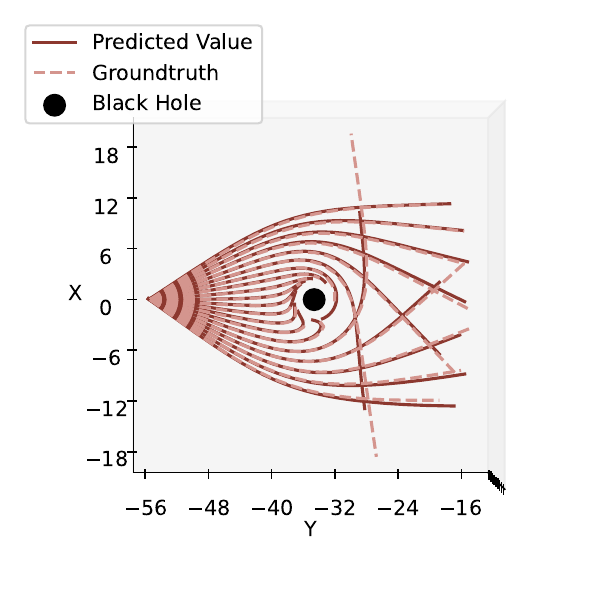}  
        \caption{}  
        \label{fig:neartraj1}
    \end{subfigure}
    \begin{subfigure}[b]{0.32\textwidth}  
        \centering
        \includegraphics[width=1\textwidth]{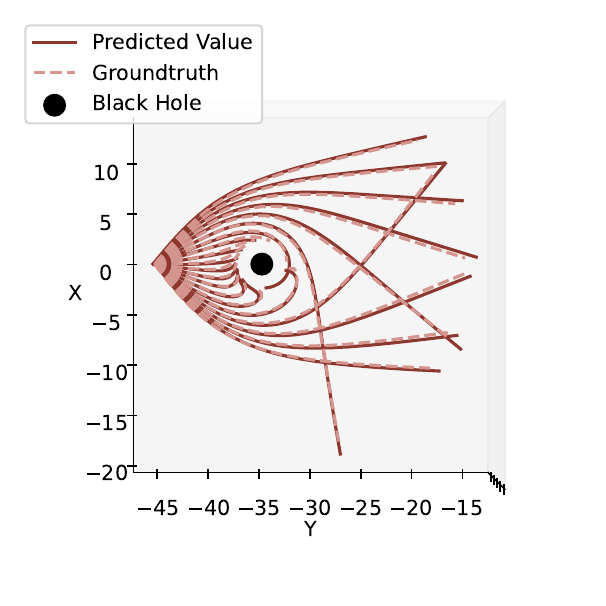}  
        \caption{}
        \label{fig:neartraj2}
    \end{subfigure}
    \begin{subfigure}[b]{0.35\textwidth}  
        \centering
        \includegraphics[width=1\textwidth]{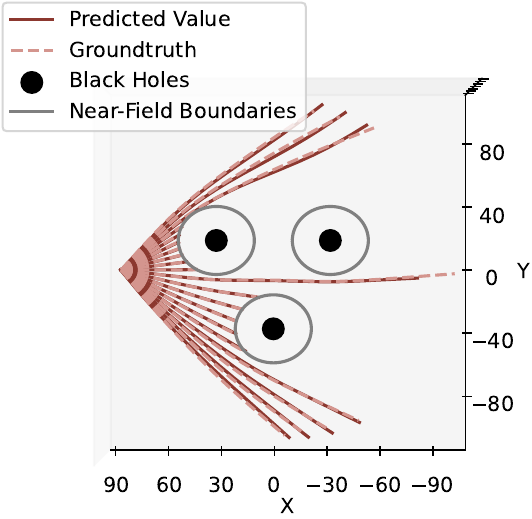}  
        \caption{}
        \label{fig:fartraj}
    \end{subfigure}
    \vspace{-20pt}
    \caption{The predicted trajectory of a light ray compared with the ground truth within near-field region \textbf{(a,b)} and far-filed region \textbf{(c)}.}\vspace{-15pt}
    \label{fig:traj}
\end{figure*}
\vspace{-6pt}

\paragraph{Color Points Sampling} 
After obtaining the ray segments, a natural idea is to sample color points within each ray segment. In our condition, the color is concentrated on the accretion disk and the sky sphere, and we record the sky sphere positions in the Ray Segments Identification once reached. Thus, we focus on sampling points on the accretion disk along geodesics. Since the accretion disk is smaller than the near-field region, we use the corresponding region's MLP to sample points. We first uniformly sample $N_{\text{coarse}}$ points for each ray segment within all near-field regions. Then, for each pair of neighboring points, we check whether the line connecting them intersects the accretion disk. For any such segment, we further sample $N_{\text{fine}}$ points along the line and identify the exact intersection with the accretion disk. Finally, we solve for the intersection point and calculate the color accordingly.
\vspace{-4pt}
\paragraph{Volume Rendering} 
We collect the colors of all points similar to the volume rendering method described in \cref{sec:bh_metric}, setting \(\Delta\lambda = 1\) due to the sparse distribution of the sampled points.

\renewcommand{\dblfloatpagefraction}{.9}

\vspace{-6pt}
\section{Experiments}

In this section we empirically evaluate \texttt{GravLensX} to show its superiority in rendering black holes for both effectiveness and efficiency. In our experiments, we use the Kerr metric to describe the spacetime around the black hole, we begin with single black hole rendering and then extend to multiple black holes. We train a 2-black-hole system and a 3-black-hole system in our experiments. For more experimental details please refer to Appendix A. For each system, we collect samples of 14,400,000 rays from each region, equivalent to nearly 7 images at a resolution of $1920 \times 1080$ for each region. During inference the velocity is calculated by finite difference for faster speed.

\begin{figure*}[htp]  
    \centering
    \includegraphics[width=\textwidth]{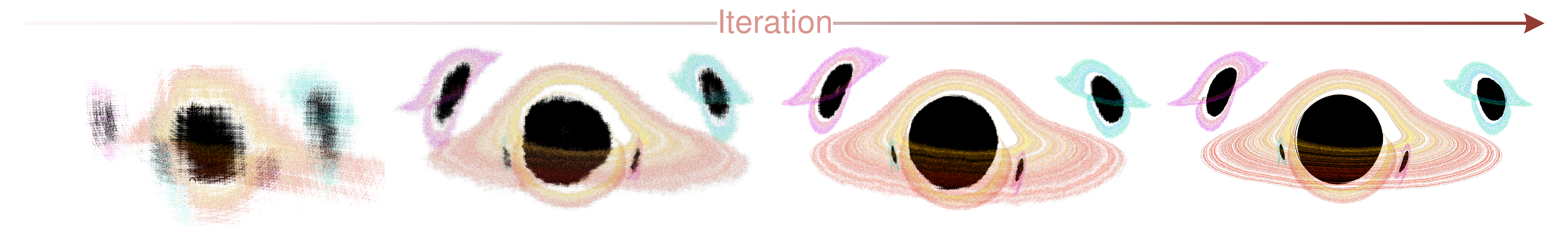}
    \vspace{-20pt}
    \caption{The rendering results of models from epochs $1$, $5$, $10$, $20$, respectively.}
    \vspace{-8pt}
    \label{fig:epoch}
\end{figure*}

\begin{figure*}[htp]  
    \centering
    \includegraphics[width=\textwidth]{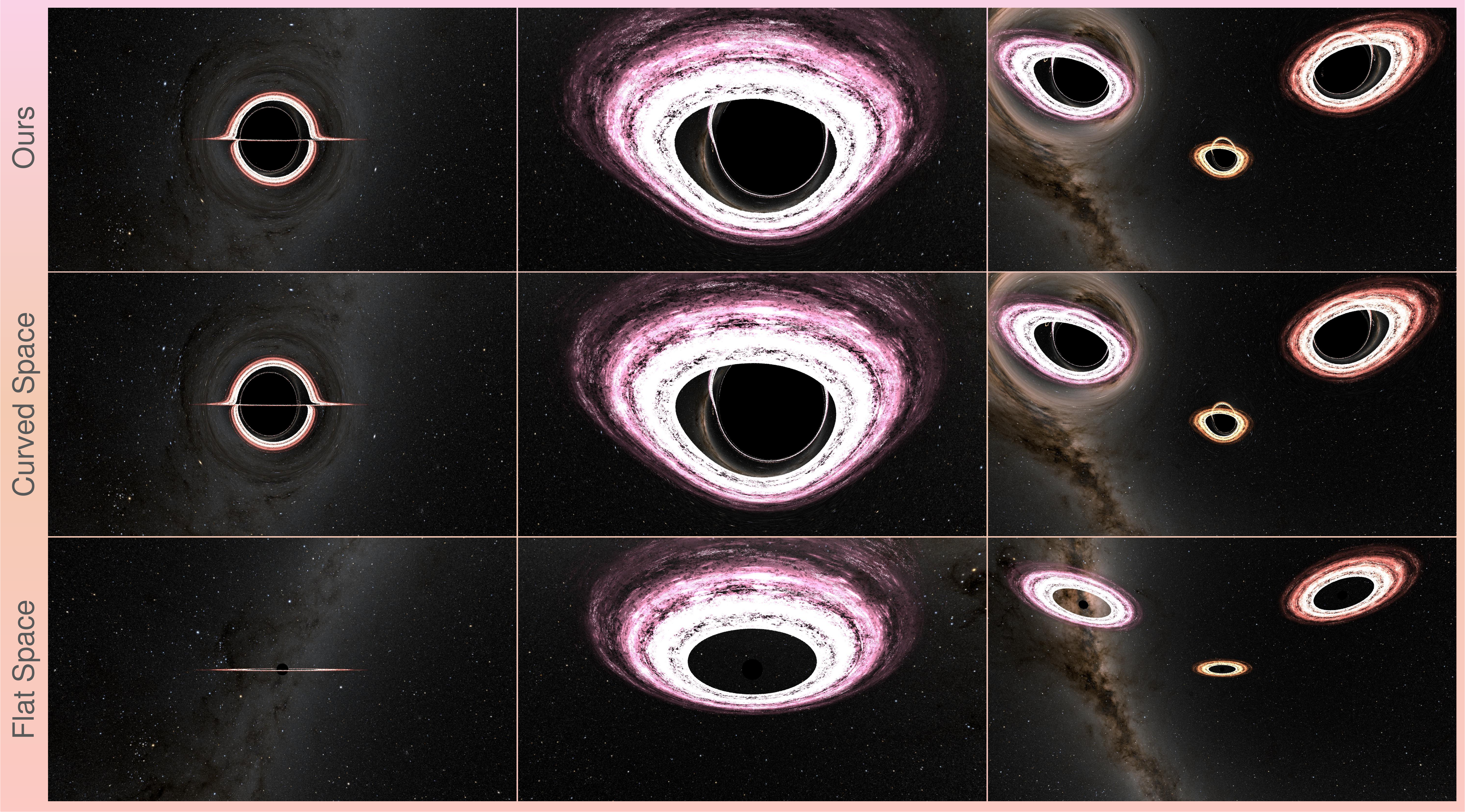}
    \vspace{-20pt}
    \caption{The rendering results including both accretion disks and sky spheres.}
    \vspace{-18pt}
    \label{fig:full_render}
\end{figure*}

\vspace{-6pt}
\subsection{Geodesic Approximation}
\subsubsection{Loss Distribution}
First we evaluate the performance of our trained neural networks compared to the ground truth on the null geodesic paths. The mean square error is leveraged as the metric.

\begin{table}
    \centering
    \caption{Quantitative comparisons between our rendering results and the Euler rendering results.}
    \label{tab:quantitative}
    \vspace{-4pt}
    \resizebox{\columnwidth}{!}{
    \begin{tabular}{cccccc} 
    \toprule\midrule
    \multirow{2}{*}{Accretion Disk} & \multirow{2}{*}{Sky Sphere} & \multicolumn{2}{c}{2 Black Holes} & \multicolumn{2}{c}{3 Black Holes}  \\ 
    \cmidrule{3-6}
                               &                             & PSNR$\uparrow$  & LPIPS$\downarrow$                     & PSNR$\uparrow$  & LPIPS$\downarrow$                      \\ 
    \midrule
    \checkmark                       & \checkmark                        & 19.47 & 0.168                     & 20.52 & 0.143                      \\
                               & \checkmark                        & 24.65 & 0.161                     & 24.98 & 0.133                      \\
    \checkmark                       &                             & 20.12 & 0.057                     & 21.70 & 0.042                      \\
    \midrule\bottomrule
    \end{tabular}
    }\vspace{-4pt}
\end{table}

As shown in \cref{fig:loss_count}, The distribution of loss and data counts are quite different for near-field MLPs and far-field MLPs. The loss of far-field MLPs are generally smaller than the loss of near-field MLPs. This is well aligned with the fact that the spacetime curvature is more significant in the near-field region, leading to more complex geodesics. Regarding the data counts, it decrease as the lambda increases for both regions, and as the data count increases, the loss tends to decrease. There's another phenomenon in specifically near-field region that when the lambda exceeds $40$, the distribution of data is quite uniform for each interval, and would last until very large lamda of around 500. This is attributed to the fact that light rays would continue to travel around the center of the black hole for long very long distances until they reach the inner boundary of the near-field region. This kind of data does not actually contributes to the rendering, and a careful selection of the inner boundary would reduce this kind of data.

\vspace{-6pt}
\subsubsection{Geodesic Visualization}
Here we compare the geodesic predicted by our MLP with the ground truth. We visualize the null geodesics of light rays casted from certain view point in the 3-black-hole system for near-field region and far-field region. As can be seen in the \cref{fig:traj}, the predicted trajectory of the light ray is quite close to the ground truth, indicating the effectiveness of our MLP in approximating the geodesic. For both systems, the error increase as the $\lambda$ increases, which is consistent with the loss distribution in \cref{fig:loss_count}. As the black hole is rotating, even the directions we sample are symmetric, the paths of the light rays are not symmetric, especially obvious in \cref{fig:fartraj}. Overall, the error of the geodesic approximation is quite small, demonstrating the effectiveness of our method. Additionally, we could see that even in the far field the light rays exhibit a slight but still noticeable curvature, highlighting the necessity of the far-field network.

\begin{figure*}[htp]  
    \centering
    \includegraphics[width=\textwidth]{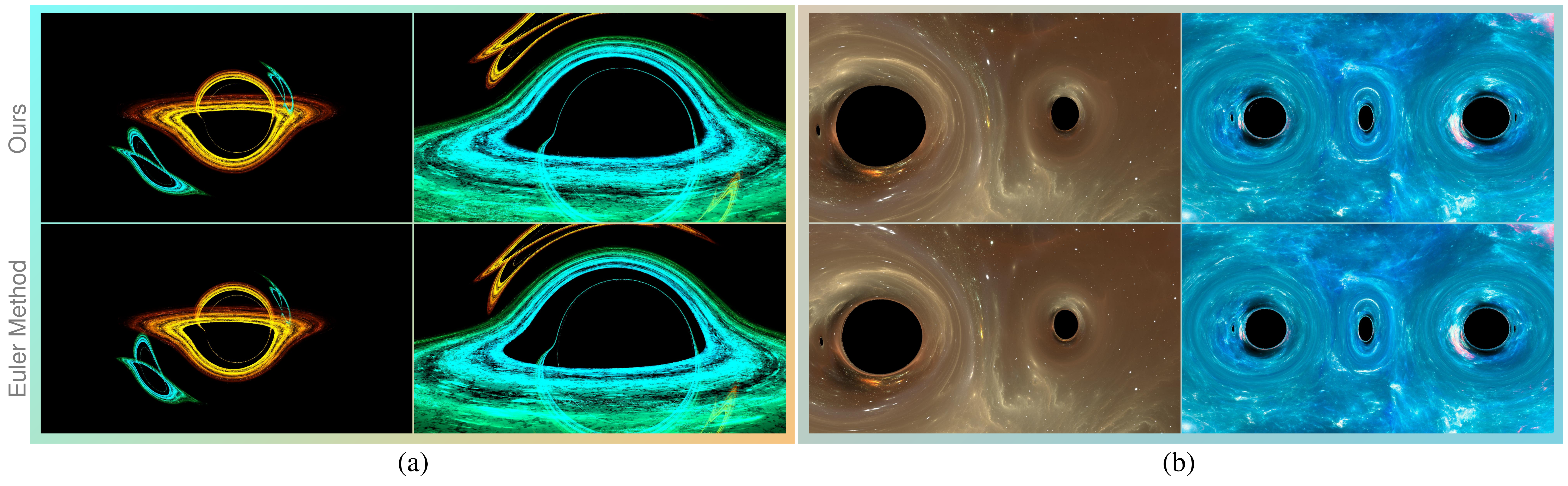}
    \vspace{-20pt}
    \caption{The separete rendering results of \textbf{(a)} accretion disk and \textbf{(b)} sky sphere.}
    \vspace{-10pt}
    \label{fig:separate_render}
\end{figure*}
\begin{figure*}[htp]  
    \centering
    \begin{subfigure}[b]{0.45\textwidth}  
        \centering
        \includegraphics[width=1\textwidth]{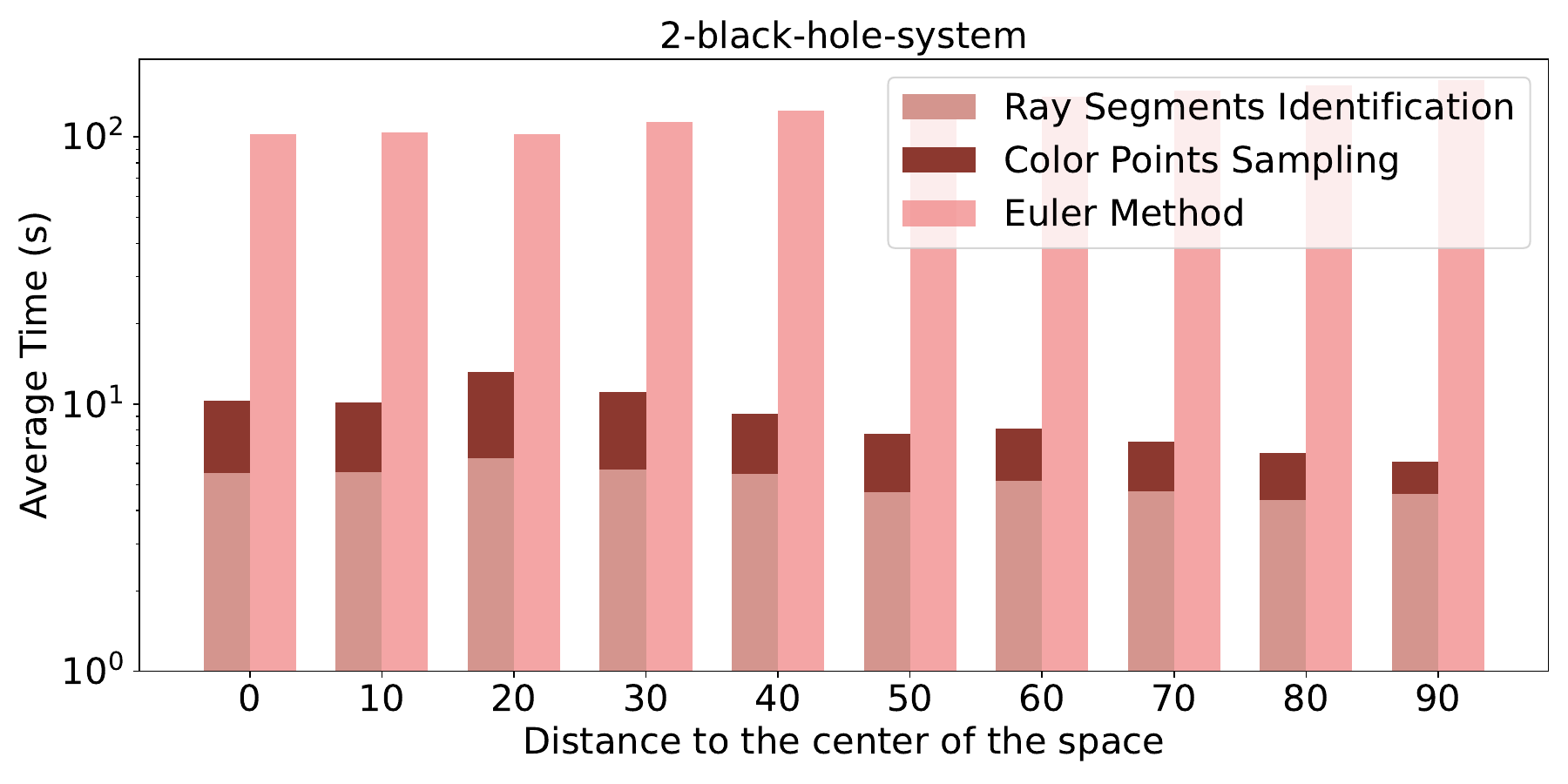}  
        \caption{}  
        \label{fig:2bhtime}
    \end{subfigure}
    \hspace{0.05\textwidth}  
    \begin{subfigure}[b]{0.45\textwidth}  
        \centering
        \includegraphics[width=1\textwidth]{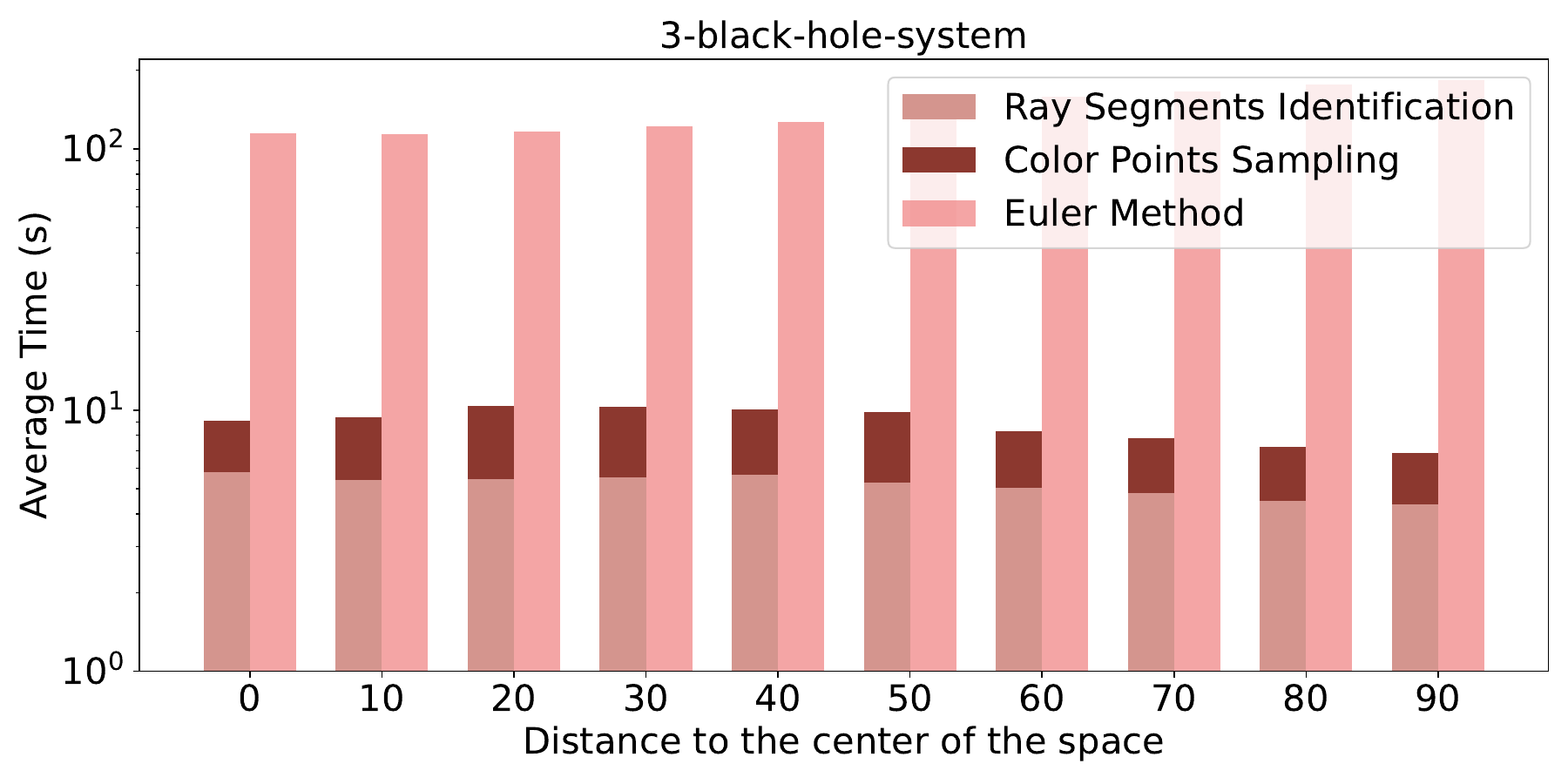}  
        \caption{}
        \label{fig:3bhtime}
    \end{subfigure}
    \vspace{-6pt}\caption{The average rendering time under different center distances for \textbf{(a)} 2-black-hole system and \textbf{(b)} 3-black-hole system.}\vspace{-6pt}
    \label{fig:time}
    \vspace{-6pt}
\end{figure*}

\vspace{-4pt}
\subsection{Rendering Results}
\vspace{-2pt}
\subsubsection{Convergence Analysis}
As shown in \cref{fig:epoch}, we visualize the rendering results of models from different training stages. It can be seen that the result of the first epoch yields a very blurry image, but the black holes and their accretion disks are still recognizable. There's no more difference between the contour of the black holes of epoch 5 to the that of later epochs, indicating that the Far-field MLP has already converged well. However the near-field MLPs need more epochs to converge, since there are still noisy color pixels around the accretion disk in the result of epoch 10.
\vspace{-8pt}
\subsubsection{Quantitative Results}
\label{sec:quantitative}
We randomly select 100 points across the entire space, and for each one, we set the viewing direction to point toward the center of a black hole. Then we render the image based on \texttt{GravLensX} along with the Euler method (ground truth) and compare their image similarity, considering different color parts: accretion disk and sky sphere. The results are shown in \cref{tab:quantitative}. Since the error of ray would cumulate as the ray goes across different regions, causing the pixels to deviate from their original positions, leading to a relatively low PSNR. However, the LPIPS of our method is quite small, implying that the images of our results are perceptually very similar to the ground truth.

\vspace{-4pt}
\subsubsection{Qualitative Results}
The comprehensive rendering results of 3-black-hole sytstem are shown in \cref{fig:full_render}. Black holes are rendered with high fidelity, and the accretion disk is clearly visible. Compared to the flat space where gravitational lensing does not exists, in curved spacetime the light rays are bent by the black holes, leading to clear gravity lensing effect in the images. We also find an interesting fact from the first column that a rotating black hole's visual surface is an ellipse shifted horizontally from its real position. This aligns well to the results presented by \citet{bohn2015does}. We also present individual rendering results for the accretion disk and sky sphere in \cref{fig:separate_render}. The differences between the image pairs, regarding the accretion disk, the sky sphere, and the shadow of the black holes, are almost indistinguishable to the naked eye. We encourage readers to view the supplementary video for additional qualitative insights.

\vspace{-4pt}
\subsection{Efficiency Analysis}

Here we compare the efficiency of \texttt{GravLensX} with the traditional Euler method. We measure the time cost of rendering a single iamge with resolution $1920\times 1080$, and we consider rendering images from different viewpoints belonging to different regions in the space. We generate 10 groups of points whose distances to the center of the space vary from 0 to 90 in dimensionless geometric units with $G=c=M=1$, and we obtain the directions using the same rule in \cref{sec:quantitative}. We  Then we record the rendering time of these points. As demonstrated in \cref{fig:time}, our approach significantly outperforms the traditional method in terms of speed, with the performance advantage becoming increasingly significant as the distance of the view point to the center grows. On average, \texttt{GravLensX} delivers a rendering speed $15\times$ faster when accounting for both the accretion disk and the sky sphere, and $26\times$ faster when focusing solely on the sky sphere. Moreover, our method is compatible with a variety of point sampling techniques—such as hierarchical schemes, empty-space skipping, proposal networks, and curvature-based heuristics—underscoring its capacity to deliver even faster rendering performance.
\vspace{-6pt}
\section{Conclusion}
\vspace{-2pt}
In this study, we have successfully demonstrated the application of neural networks for the rendering of black holes in curved spacetime, a process traditionally constrained by high computational demands. Our approach not only leverages the inherent physical laws of general relativity but also significantly optimizes the computational efficiency of simulating gravitational lensing effects.

Experimental validations demonstrate that our learning-based approach effectively approximates the geodesic paths of light in both near-field and far-field scenarios, producing high-quality black hole images in significantly less time than traditional methods. This innovation could have a profound impact on the study of black holes and other gravitational phenomena, providing researchers with a powerful tool for visualizing and analyzing these phenomena in unprecedented detail.

\section{Acknowledgement}
We thank William Throwe for his insightful discussion and the anonymous reviewers for their valuable feedback. This work was supported in part by the National Natural Science Foundation of China under Grants 62073066, and in part by 111 Project under Grant B16009.

{
    \small
    \bibliographystyle{ieeenat_fullname}
    \bibliography{main}
}
\appendix

\setcounter{table}{0}   
\setcounter{figure}{0}
\setcounter{section}{0}
\setcounter{equation}{0}
\renewcommand{\thetable}{A\arabic{table}}
\renewcommand{\thefigure}{A\arabic{figure}}
\renewcommand{\thesection}{A\arabic{section}}
\renewcommand{\theequation}{A\arabic{equation}}

\onecolumn
{\centering
  \LARGE\bfseries Appendix\par
}
\section{Additional Experiments}
\begin{figure}[htp]  
    \centering
    \begin{subfigure}[b]{0.39\columnwidth}  
        \centering
        \includegraphics[width=1\textwidth]{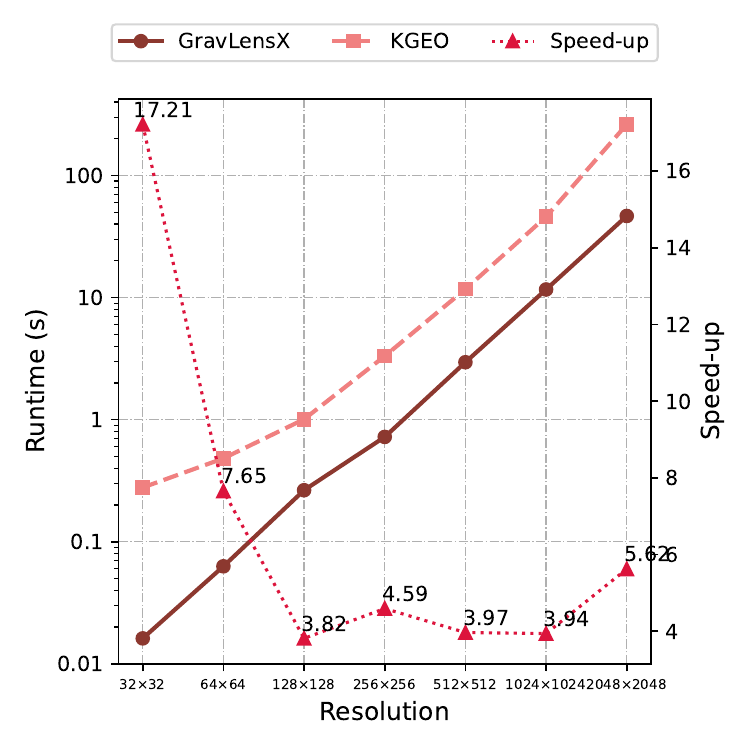}  
    \end{subfigure}
    \begin{subfigure}[b]{0.39\columnwidth}  
        \centering
        \includegraphics[width=1\textwidth]{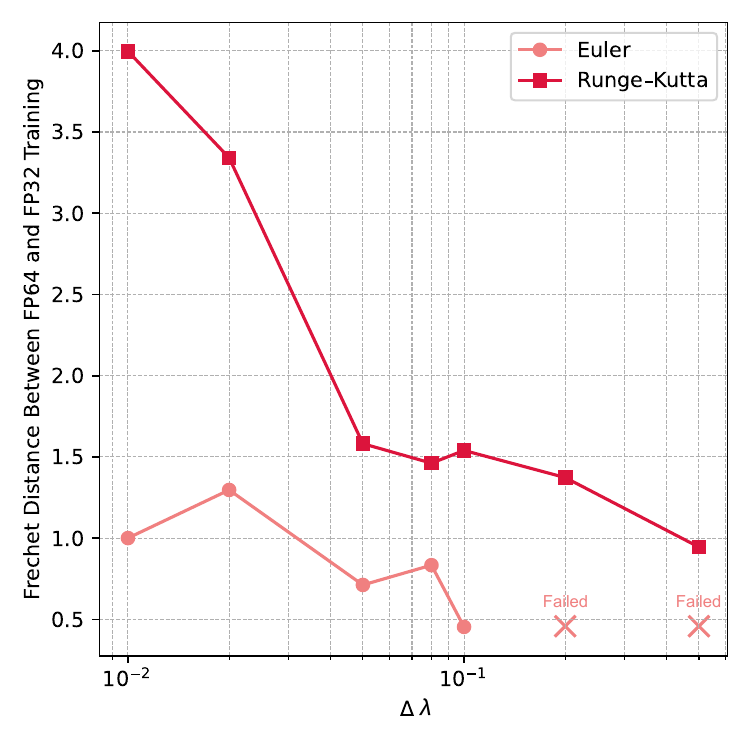}  
    \end{subfigure}
    \caption{\textbf{(Left)} Runtime comparison between kgeo and GravLensX. 100 points are sampled for each ray. \textbf{(Right)} Frechet distances between paths using FP32 and FP64.}
    \label{fig:kgeo}
\end{figure}

kgeo is an implementation of \citet{gralla2020null}'s work that provides a fast and accurate solution to the null geodesic equations in Kerr spacetime with elliptic integrals. For single-black-hole scenario, we benchmarked kgeo on a Xeon Gold 5218 CPU by launching 16 parallel processes to fully utilize all cores, and compared this to \texttt{GravLensX} on a single NVIDIA RTX 3090 GPU. \texttt{GravLensX} yields speed-ups of roughly $3.8\times$ to $17.2\times$ (see \cref{fig:kgeo} Left). More importantly, for the multi-black-hole scenarios we consider in our paper leveraging the superposed Kerr metric, \citet{gralla2020null}'s theory cannot be utilized.
\section{Implementation Details}
\label{sec:details}
We set $N_{\text{coarse}}=N_{\text{fine}}=10$. $R^{bh}$ is set as $20$. Detailed parameters for generating data are listed in \cref{tab:param}. Taichi~\citep{hu2019taichi} and PyTorch~\citep{Ansel_PyTorch_2_Faster_2024} are used to render the black hole systems and train the neural networks.
\begin{table}[h]
\centering
\caption{Parameters of generating geodesic data.}
\label{tab:param}
\begin{tabular}{ccccccc} 
\toprule
System        & $P^{\text{bh}}$   & $M^{\text{bh}}$ & $A^{\text{bh}}$ & $l^{\text{in}}$ & $l^{\text{out}}$ & in-black-hole radius \\
\midrule
2 Black Holes & (-30, 0, 0), (30, 0, 0)   & 1, 1    & 1, 1         & 1.6      & 100  & 1.8  \\                
3 Black Holes & $(30, 10\sqrt{3}, 0), (-30, 10\sqrt{3}, 0), (0, -20\sqrt{3}, 0)$ & 1, 1, 1              & 1, 1, -1             & 2.2                  & 100 &  2.25  \\                 
\bottomrule
\end{tabular}
\end{table}

As for the training process, we use MLPs with 12 hidden layers and residual connections, employing SoftPlus as the activation function. For the near field network, each layer contains 200 neurons. For the far field network, each layer has 128 neurons in the 2-black-hole system and 200 neurons in the 3-black-hole system. The learning rate is set as $0.001$ with default Adam optimizer. We trained the far-field network on 4 RTX 3090 GPUs and the near-field network on single RTX 3090. The time for generating data and training the model on RTX 3090 is around 13.8 GPU hours for every near-field NN and 33.5 GPU hours for the far-field NN. For the 3-black-hole scenario, this is equivalent to rendering a 63 minutes video in 30 FPS with Euler method. For more complex systems, the training time grows linearly with the number of black holes. We did not investigate much on the training acceleration, and we believe the training time can be largely reduced by using more advanced optimizers, model structures, etc. 
\section{Superposed Black Hole Metric}
\label{sec:metric}
Black hole metric, e.g. Schwarzschild metric~\citep{schwarzschild1916gravitationsfeld}, is the solution to Einstein field equations under certain assumptions. We primarily utilize the Kerr metric~\citep{kerr1963gravitational} that describes a rotating black hole to characterize the black hole system, as adopted by earlier studies~\citep{james2015gravitational, bohn2015does, porter2024parameter}. Kerr metric in Cartesian coordinates is defined as 
\begin{equation}
    \label{eq:kerr_metric}
    \mathrm{d} s^2=  -\mathrm{d} t^2+\mathrm{d} x^2+\mathrm{d} y^2+\mathrm{d} z^2 +\frac{2 m r^3}{r^4+a^2 z^2} \left(\mathrm{~d} t+\frac{r(x \mathrm{~d} x+y \mathrm{~d} y)}{a^2+r^2}+\frac{a(y \mathrm{~d} x-x \mathrm{~d} y)}{a^2+r^2}+\frac{z}{r} \mathrm{~d} z\right)^2,
\end{equation}
where the Boyer-Lindquist radius $r$ is implicitly given by
\begin{equation}
    \label{eq:r}
    x^2+y^2+z^2=r^2+a^2\left(1-\frac{z^2}{r^2}\right).
\end{equation}
$m$ and $a$ are the mass and the angular momentum of the black hole, respectively. $t$, $x$, $y$, and $z$ are the coordinates in the spacetime. $ds$ represents an infinitesimal spacetime interval, which captures the separation between two nearby events in spacetime. With \cref{eq:kerr_metric}, The metric tensor for these four coordinates can be calculated as
\begin{equation}
    g_{a b}=\eta_{a b}+\frac{2 m r^3}{r^4+a^2 z^2} \ell_a \ell_b,
\end{equation}
with
\begin{equation}
    \label{eq:ell}
    \ell=\left[1, \frac{r x+a y}{r^2+a^2}, \frac{r y-a x}{r^2+a^2}, \frac{z}{r}\right]^\mathsf T,
\end{equation}
where $\eta_{a b}$ is the Minkowski metric tensor. We use an affine parameter $\lambda$ to parameterize the geodesic as done by \citet{d2018electromagnetic} and \citet{porter2024parameter}. Let $p=[t, x, y, z]^{\mathsf{T}}$ denote the spacetime coordinates. The geodesic equation can be written as
\begin{equation}
    \label{eq:geodesic}
    \frac{\mathrm{d}^2 p^\mu}{\mathrm{d} \lambda^2}+\Gamma_{\alpha \beta}^\mu \frac{\mathrm{d} p^\alpha}{\mathrm{d} \lambda} \frac{\mathrm{d} p^\beta}{\mathrm{d} \lambda}=0
\end{equation}
in which $\lambda$ is the affine parameter, ensuring the path is parameterized in a way that preserves the affine properties of the curve. $\Gamma_{\alpha \beta}^\mu$ is the Christoffel symbol describing how coordinates change in curved spacetime, which can be calculated from the metric tensor as
\begin{equation}
    \Gamma_{\alpha \beta}^\mu=\frac{1}{2} g^{\mu \nu}\left(\frac{\partial g_{\nu \alpha}}{\partial p^\beta}+\frac{\partial g_{\nu \beta}}{\partial p^\alpha}-\frac{\partial g_{\alpha \beta}}{\partial p^\nu}\right),
\end{equation}
where $g^{\mu \nu}$ is the inverse of the metric tensor, satisfying
\begin{equation}
  g^{\mu\nu} g_{\nu\alpha} = \delta^\mu_\alpha= 
  \begin{cases} 
  1 & \text{if } \mu = \alpha \\ 
  0 & \text{if } \mu \neq \alpha 
  \end{cases}
\end{equation}  
Calculating the partial derivative term in the Christoffel symbol is the most computationally expensive part in solving the geodesic equation. In order to calculate the second order derivative of each coordinate w.r.t. $\lambda$ in \cref{eq:geodesic}, we need to obtain the derivative of the time coordinate w.r.t. $\lambda$, $\frac{\mathrm{d} t}{\mathrm{d} \lambda}$. According to the general relativity, light travels along null geodesics, indicating that the derivatives of the four coordinates $p$ satisfy
\begin{equation}
    \label{eq:null_geodesic}
    g_{\mu \nu} \frac{\mathrm{d} p^\mu}{\mathrm{d} \lambda} \frac{\mathrm{d} p^\nu}{\mathrm{d} \lambda}=0.
\end{equation}
Solving this equation yields $\frac{\mathrm{d} t}{\mathrm{d} \lambda}$. By discretizing the $\lambda$ and iterately solve the spacetime coordinate at each lambda step.

\section{Proof for $l^{\text{straight}} \leq l^{\text{geodesic}}$ and Its Convergence}
\label{sec:proof}
\begin{figure}[htp]  
    \centering
    \begin{subfigure}[b]{0.25\textwidth}  
        \centering
        \includegraphics[width=1\textwidth]{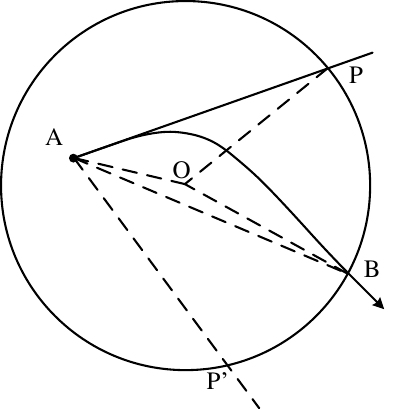}  
        \caption{}
        \label{fig:lambda_inner}
    \end{subfigure}
    \hspace{0.1\textwidth}
    \begin{subfigure}[b]{0.25\textwidth}  
        \centering
        \includegraphics[width=1\textwidth]{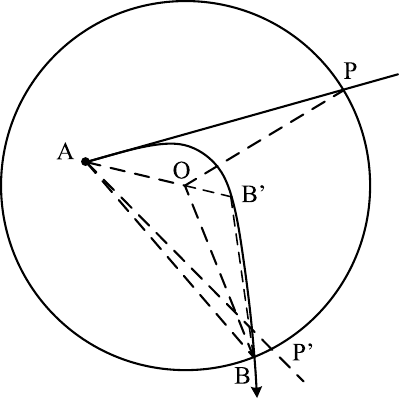}  
        \caption{}  
        \label{fig:lambda_outer}
    \end{subfigure}
    \hspace{0.1\textwidth}
    \begin{subfigure}[b]{0.25\textwidth}  
        \centering
        \includegraphics[width=1\textwidth]{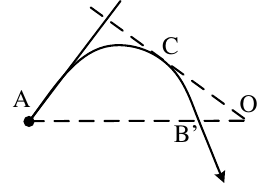}  
        \caption{}  
        \label{fig:lambda_false}
    \end{subfigure}
    \caption{Different cases of intersections.}
\end{figure}

Suppose a ray is casted from point $A$ to point $B$ in the near-field region with radius $R$ centered at $O$, where $B$ is its first crossing point to the boundary. The nonzero acceleration $\mathrm{d}^2r/\mathrm{d}\lambda^2 < 0$ caused by gravity bends the ray towards the center. $d$ is the initial direction of the ray and $P$ is the point where the straight line from position $P$ with direction $d$ intersects the boundary.
\begin{proposition}
\label{prop:bound}
Let $L_{AB}$ be the length of the ray from point $A$ to point $B$. The following inequality holds:
\begin{equation}
    L_{AB} \geq |AP|.
\end{equation} 
\end{proposition}
\begin{proof}
Consider the plane expanded by points $A$, $O$, and $B$. Let $P'$ be the reflection of $P$ across line $AO$. 
\begin{description}
\item[Case 1: Point $B$ is on the arc $\wideparen{PP'}$] As shown in \cref{fig:lambda_inner}, through cosine law, we have
\begin{equation}
    \begin{aligned}
        |AP|&=R^2 + |AO|^2 - 2R |AO| \cos(\angle AOP),\\
        |AB|&=R^2 + |AO|^2 - 2R |AO| \cos(\angle AOB).\nonumber
    \end{aligned}
\end{equation}
Clearly, $\angle AOB \geq \angle AOP$, which implies $L_{AB} \geq |AB| \geq |AP|$, where the latter equality only holds when $B$ is on $P$ or $P'$.

\item[Case 2: Point $B$ is not on the arc $\wideparen{PP'}$]
As shown in \cref{fig:lambda_outer}, extending line $AO$ intersects the trajectory $AB$ firstly at point $B'$. Assume the crossing point $B'$ lies on $AO$ (see \cref{fig:lambda_false}). Draw line $OC$ through point $O$ such that $OC$ is tangent to curve $AB'$, with $C$ being the point of tangency. Denote the position and speed of the curve at point $C$ as $p_c =(x_c, y_c)^\mathsf{T}$, $v_c =(-k\cdot x_c, -k\cdot y_c)^\mathsf{T}$ in which $k$ is a constant value, as its direction is parallel to the tangent line across the coordinates origin. Also, the acceleration is parallel to the direction $v_c$, thus the trajectory of the ray after $C$ should be a straight line pointing towards $O$.
This contradicts our assumption that the curve would continue bending until it reaches $AO$ at $B'$, leading to the conclusion that $B'$ cannot lie on $AO$ and should instead intersect the extension of line $AO$ towards the O end. Now we have
\begin{equation}
    L_{AB} = L_{AB'} + L_{B'B} \geq |AO| + |OB'| + |B'B| \geq |AO| + |OB| \geq |AP|.\nonumber
\end{equation}
\end{description}
\end{proof}
In our method, we iteratively step over the ray path to conduct ray tracing. We denote the start point and direction as $A_0$ and $d_0$. At each step, we calculate the length of $A_0P_0$, where $P_0$ is the point of intersection between the boundary and the straight line cast from $A_0$ in the direction of $d_0$. Then we step with length $s_0=|A_0P_0|$ on the ray path and obtain next point $A_1$ and direction $d_1$. This iteration process is conducted until the point is sufficiently near to the boundary, i.e., sufficiently close to $B$.

\begin{proposition}
    Given a point on a ray path influenced by the gravity in the near-field, we have
    \begin{equation}
        \lim_{T\to \infty}\sum_{t=0}^T s_t = L_{A_0B}.
    \end{equation}
\end{proposition}
\begin{proof}
    Through our definition, we have
    \begin{equation}
        s_t \geq 0, \qquad\forall t \in \mathbb{N}.\nonumber
    \end{equation}
    By \cref{prop:bound}, it holds that
    \begin{equation}
        \sum_{t=0}^{T}s_t \leq L_{A_0B} < +\infty, \qquad\forall T \in \mathbb{N},\nonumber
    \end{equation}
    indicating that $\sum_{t=0}^{T}s_t$ converges, and 
    \begin{equation}
    \lim_{t\to \infty} s_t=0. \nonumber
    \end{equation}
    Assume $r:=L_{A_0B} - \sum_{t=0}^{\infty}s_t > 0$, pick $B'$ on the path so that $L_{B'B}=r$. Clearly $A_t$ converges to $B'$ as $t$ increases, and it holds that
    \begin{equation}
        \lim_{t\to \infty} s_t = s_{B'},\nonumber
    \end{equation}
    in which $s_{B'}$ is the distance from $B'$ to the boundary across its tangent line. As $B'$ does not overlap with $B$, 
    \begin{equation}
        \lim_{t\to \infty} s_t = s_{B'} > 0,\nonumber
    \end{equation}
    this contradicts to the aforementioned condition, therefore $r>0$ does not hold, implying that
    \begin{equation}
        \sum_{t=0}^{\infty}s_t=L_{A_0B}.\nonumber
    \end{equation}
\end{proof}

\end{document}